\documentclass[a4paper]{article}

\usepackage{graphicx}
\usepackage{amsthm}
\usepackage{mathptmx}
\usepackage{amssymb}
\usepackage{version}
\usepackage{subfigure}
\usepackage{xcolor}
\usepackage{tikz-qtree}
\usepackage{graphicx}

\newtheorem{Example}{Example}

\newtheorem{remark}{Remark}

\newtheorem{lemma}{Lemma}
\newtheorem{theorem}{Theorem}

\newcommand{\D}{\mathcal D}

\newcommand{\dist}{dist}
\begin{document}

\providecommand{\keywords}[1]
{
  \small	
  \textbf{\textit{Keywords:}} #1
}

\title{A new distance based on minimal absent words and  
applications to biological sequences.} 


\author{Giuseppa Castiglione, Jia Gao, Sabrina Mantaci, Antonio Restivo}
\date{%
    Dipartimento di Matematica e Informatica, Universit{\`a} di Palermo, Italy\\
\texttt{\{giuseppa.castiglione,jia.gao, sabrina.mantaci,antonio.restivo\}@unipa.it}}

\maketitle

\begin{abstract}
 A minimal absent word of a sequence $x$, is a sequence $y$ that is not a factor of $x$, but all of its proper factors are factors of $x$ as well. The set of minimal absent words uniquely defines the sequence itself. In recent times minimal absent words have been used in order to compare sequences. In fact, to do this, one can compare the sets of their minimal absent words. Chairungasee and Crochemore in \cite{Croche2012} define a distance between pairs of sequences $x$ and $y$, where the symmetric difference of the sets of minimal absent words of $x$ and $y$ is involved. Here, we consider a different distance, introduced in \cite{CastiglioneMR20}, based on a specific subset of such symmetric difference that, in our opinion, better capture the different features of the considered sequences. We show the result of some experiments where the distance is tested on a dataset of genetic sequences by 11 living species, in order to compare the new distance with the ones existing in literature.
\end{abstract}

\section{Introduction}
It is well-known that sequence comparison finds many applications in comparative genomics for the study of evolutions, for building phylogenies, for comparing viruses genomes. In the context of alignment-free methods for sequence comparison the concept of minimal absent word has often been used in recent years. The method is based on the idea that the negative information well represent the sequence itself, hence two sequences can be  compared by comparing the relative sets of minimal absent words. 
The advantages of this approach are that the set of minimal absent words uniquely characterizes the sequence (cf. \cite{MignosiRS01}), the number of minimal absent words of a sequence of length $n$ is linear in $n$ (cf. \cite{CMR98}), and they can be computed in linear time \cite{Croche2018}. As a consequence, it is possible to compare two sequences in time proportional to their lengths. An experimental study of different distance measures based on minimal absent words to analyze similarity/dissimilarity of sequences has been carried out in \cite{RACR2016}. In \cite{Croche2012} Chairungsee and Crochemore introduced a measure of similarity between two sequences $x$  and $y$ making use of a {\em length-weighted index} on the symmetric difference  $M(x) \triangle M(y)$  of the sets of minimal absent words  $M(x)$ and $M(y)$  of  $x$  and  $y$, respectively. In the same paper the authors propose to evaluate the length-weighted index on a {\em sample set}, i.e. the subset of  $M(x) \triangle M(y)$ of words of limited length. Further developments and an extension of the ideas of \cite{Croche2012} can be found in \cite{Croche2018}.


In \cite{CastiglioneMR20}  new similarity measures between sequences based on minimal absent words have been introduced with the aim to deepen a theoretical comparison with the measures in \cite{Croche2012} and \cite{EH1989}.


The flaw of the distance in \cite {Croche2012} is that the set $M(x) \triangle M(y)$ could contain  words that are absent both in $x$ and in $y$, although they are minimal only for one of them.  In our opinion, if the aim is to distinguish $x$ and $y$ it is not appropriate to consider such words.
Hence, we propose to evaluate the length-weighted index on the sample set
$$\D(x,y) = (F(x) \cap  M(y)) \cup (F(y) \cap M(x)),$$
where $F(x)$ denotes the set of factors of $x$.
The set  $\D(x,y)$  contains words that are absent in one of the two words ($x$ or $y$), but that are factors of the other one. In our proposal only the words of $\D(x,y)$ really contribute to distinguish $x$ and $y$. From the algebraic point of view, the set  $\D(x,y)$ is the base of the ideal generated by $M(x)\triangle M(y)$, hence the choice of $\D(x,y)$ as a sample set corresponds to eliminate those words of  $M(x) \triangle M(y)$ that have a proper factor in the same set. For this reason, in general, $\D(x,y)$ has far fewer elements than $M(x)\triangle M(y)$ and $\D(x,y)$ contains words among the shortest of $M(x)\triangle M(y)$. This choice, from the practical point of view, has a potential advantage in terms of computation time.


In the first part of the paper we recall some definitions and main results about the distance defined in \cite{CastiglioneMR20} and we make an experimental study to analyze the similarity/dissimilarity of different sequences using the two distances based on the notion of absent words. The main scope is to propose our method to the community as an instrument that works well in the sequence comparison and phylogeny reconstruction. Hence, we evaluate our experimental results following the methology in \cite{Croche2012}, \cite{LIU2005} and \cite{RACR2016}. 

We show the result of some experiments where the distance is tested on a data set of genetic sequences by 11 living species. We compare the cardinality of $\D(x,y)$ and the cardinality of $M(x) \triangle M(y)$ and then we compare the sum of the lengths of their elements, respectively. 
Finally, we use the two distance measures analyzed to reconstruct phylogenetic trees using Unweighted Pair Group Method with Arithmetic mean (UPGMA) and Neighbor Joining (NJ) methods and we conclude with some considerations. 

\section{Definitions and notations}\label{sec:def}

Let $\Sigma$ be a finite alphabet and $\Sigma^*$ the set of the words over $\Sigma$. The set $\Sigma^*$ is the free monoid generated by $\Sigma$ with respect to the word concatenation and with the empty word $\epsilon$ the unit element. If $X\subset \Sigma^*$, we denote by $Card(X)$ its cardinalty, i.e. the number of its elements, whereas we denote by $s(X)$ the {\em total lengt} of $X$. i.e.: $$s(X)= \sum_{u \in X} \vert u \vert.$$
where $\vert u\vert$ denotes the length of the word $u$.
A set $I \subseteq \Sigma^*$ is said to be a ({\em two-sided}) {\em ideal} of $\Sigma^*$ if for each $u \in I$ and $v \in \Sigma^*$ the two concatenations $uv, vu \in I$, i.e. $I=\Sigma^* I \Sigma^*$. The {\em base} of the ideal $I$ is the minimal set $B$ (with respect to the set inclusion) such that $I=\Sigma^*B\Sigma^*$.

Let $v$ be a word of $\Sigma^*$, we say that $u$ is a {\em factor} of $v$ if there exist $z,w \in \Sigma^*$ such that $v=zuw$; if $z=\epsilon$ (resp. $w=\epsilon$) we say that $u$ is a {\em prefix} (resp. a {\em suffix}) of $v$; if $u \neq v$ we say that $u$ is a {\em proper factor} of $v$. If $u$ is a prefix of $v$, i.e. $v=uz$, we denote $u^{-1}v=z$. In what follows we denote by $F(v)$ the set of factors of $v$. We say that a word $u$ {\em occurs} in $v$ if it is a factor of $v$.


\medskip
A word $u$ is an {\em absent word} of $v$ if it does not occur in $v$. An absent word is a {\em minimal absent word} if all its proper factors occur in $v$. We denote by $M(v)$ the set of minimal absent words of $v$.

\begin{Example}
Let $v=abaabab$ then $M(v)=\{aaa, aabaa, baba, bb\}$.
\end{Example}

A language $L\subseteq \Sigma^*$ is called {\em factorial} if it contains all the factors of its own words, whereas it is called {\em antifactorial} if no word in the language is a proper factor of another word in the language. In particular, for any word $v\in \Sigma^*$, $F(v)$ is a factorial language and $M(v)$ is antifactorial.

\medskip 

Remark that the complement of $F(v)$ (i.e. the set of the words that are not factors of $v$) is an ideal of $\Sigma^*$ and $M(v)$ is its base. This allows to establish a duality between the sets $F(v)$ and $M(v)$ given by the relations (cf. \cite{CMR98}):
$$F(v)=\Sigma^* \setminus \Sigma^* M(v) \Sigma^*,$$ $$M(v)=\Sigma F(v) \cap F(v)\Sigma \cap (\Sigma^* \setminus F(v)).$$

This last relation comes from the remark that if $v\in \Sigma^*$, the word $u=a_1a_2\cdots a_n$, with $a_i\in \Sigma$ is a  minimal absent word of $v$ iff  $u \notin F(v)$ and $a_1\cdots a_{n-1}$, $a_2 \cdots a_n \in F(v)$.

The UPGMA (Unweighted Pair Group Method with Arithmetic Mean) is a technique to reconstruct the phylogenetic tree for a set of taxa given the distance matrix (cf. \cite{Sung2009}). It is based on the principle that similar taxa should be close in the phylogenetic tree, hence it builds the tree by clustering similar taxa. The NJ (Neighbor-Joining) method is based on the principle of  finding pairs of operational taxonomic units that minimize the total branch length at each stage of clustering of unit, starting with a starlike tree (cf. \cite{Saitou1987TheNM}).

\section{Similarity measures based on sets of absent words}\label{sec:similarity}

The idea to measure similarity by minimal absent words is based on the intuition that two words, $x$ and $y$, are as more distant as bigger is the set of the non common absent words and as shorter are the words in it. 

\medskip
This idea is formalized in a paper by Chairungsee and Crochemore \cite{Croche2012} where the notion of length weighted index of a set is used in order to define a similarity/dissimilarity measure of two sequences. The {\em length weighted index} is defined as the measure that associate to a set $X\subseteq \Sigma^*$ the quantity:
$$\mu(X)=\sum_{w\in X} \frac{1}{|w|^2}.$$\label{measure}

This measure is used in \cite{Croche2012}  in order to define a distance between two words $x$ and $y$, by taking the set  $X=M(x)\triangle M(y)$, where $\triangle$ denotes the symmetric difference operator between two sets. Therefore the distance is denoted and defined as:

$$\dist(x,y)=\mu(M(x)\triangle M(y))=\sum_{w\in M(x) \triangle  M(y)}\frac{1}{|w|^2}$$

\begin{Example}\label{ex-dist}
Let $x=cbaabdcb$ and $y=abcba$.
Then, $$M(x)=\{ac, ad, bb, bc, ca, cc, cd, da, db, dd, aaa, aba, bab, cbd, dcba\},$$
$$M(y)=\{aa, ac, bb, ca, cc, aba, bab, cbc, d\},$$ $$M(x)\triangle M(y)=\{d, aa, ad, bc, cd, da, db, dd, aaa, cbd, cbc, dcba\},$$
Therefore $$\dist(x,y)=
1+\frac{7}{4}+\frac{3}{9}+\frac{1}{16}=\frac{453}{144}.$$
\end{Example}
We remark that the distance between $x$ and $y$ is bigger when $M(x) \triangle  M(y)$ contains shorter words. In fact, longer minimal absent words do not substantially affect the distance $\dist(x,y)$. This is the reason why in \cite{Croche2012} the authors propose to ignore from the set $M(x)\triangle M(y)$ those words having length longer than a fixed threshold. In this way a great advantage is gained in terms of computation time.
Then, if $M_l(x)$ denotes the set of minimal absent words having length smaller then or equal to $l$, it is defined:

$$\dist_l(x,y)=\mu(M_l(x)\triangle M_l(y))=\sum_{w\in M_l(x) \triangle  M_l(y)}\frac{1}{|w|^2}$$




In \cite{CastiglioneMR20} a different distance  also based on the measure $\mu$ is considered, but applied to a subset of $M(x)\triangle M(y)$ that better captures the difference between two words. Moreover, by considering this subset, the requirement of having words with limited length is undirectely satisfied. This subset of $M(x)\triangle M(y)$ is in fact made of 
those factors of $x$ that are minimal absent words for $y$ and viceversa. In other terms, we want the comparison of the two sequances $x$ and $y$ not to be influenced by those minimal absent words of $y$ that do not occur in $x$. This idea is formally described as follows.

For all $x,y\in \Sigma^*$, we define the set:
$$D(x \leftarrow {y})=F(x)\cap M(y)$$ i.e. the set of minimal absent words of $y$ that are factors of $x$.

Given two words $x$ and $y$ we can define $$\mathcal{D}(x,y) =D(x \leftarrow {y})\cup D(y \leftarrow {x})=(F(x)\cap M(y))\cup (F(y)\cap M(x)).$$

\begin{Example}\label{ex-Dxy}
Let $x=cbaabdcb$ and $y=abcba$ as in Example \ref{ex-dist}. we have that:
$$M(x)\triangle M(y)=\{d, aa, ad, bc, cd, da, db, dd, aaa, cbd, cbc, dcba\},$$
$$\mathcal{D}(x,y)=\{d, aa, bc\}.$$
 Remark that the word $cd$, for instance, is absent both in $x$ and in $y$ (although not minimal in $y$) so, in some way, it represents a common property of the two words, and it should not be considered its contribution to the distance. The same holds for the words $ad, da, db, dd, aaa, cbd, cbc$, and $dcba$. On the other hand, the word $d$, for instance, is a minimal absent word in $y$, but occurs in $x$ and therefore discriminates the two words. Viceversa, the word $aa$ is minimal absent in $y$ but occurs in $x$ i.e. it also contributes to their dissimilarity. 
In this example, the cardinality of the set $\D(x,y)$ is smaller than the one of the symmetric difference, and the words in $\D(x,y)$ are among the smallest in $M(x)\triangle M(y)$.
\end{Example}

The following two properties have been proved in \cite{CastiglioneMR20}.

\begin{remark}
For all $x,y\in \Sigma^*$,
\begin{enumerate}
    \item\label{dzero} $\D(x\leftarrow y)=\emptyset$ if and only if $x\in F(y)$;
    \item\label{dcalzero} $\D(x,y)=\emptyset$ if and only if $x=y$.
    \end{enumerate}
\end{remark}

In what follows we recall some algebraic properties of $\D(x,y)$ also in relation with $M(x)\triangle M(y)$.

\begin{lemma}\label{lm:inclusion}
For all $x,y\in \Sigma^*$,
    $\mathcal{D}(x,y) \subseteq M(x)\triangle M(y)$.
\end{lemma}
\begin{proof}
Let $z\in \mathcal{D}(x,y)$. Then either $z\in F(x)\cap M(y)$ or $z \in F(y)\cap M(x)$. Suppose $z\in F(x)\cap M(y)$ (the proof for the other case is symmetric). Then $z\in M(y)$ and $z\not \in M(x)$, therefore $z \in M(y)\backslash M(x)$. \end{proof}

\begin{lemma}\label{lm:antifactorial}
For all $x,y\in \Sigma^*$,
$\D(x,y)$ is antifactorial.
\end{lemma}
\begin{proof}
First, note that $\epsilon \notin \D(x,y)$ because it is never an absent factor i.e. $\epsilon \notin M(x)$ and $\epsilon \notin M(y)$. We have to prove that for all $z\in \D(x,y)$ none of its factor is in $\D(x,y)$.
By contradiction suppose there exists a $z\in \D(x,y)$ and a factor $z'$ of $z$ such that $z'\in \D(x,y)$. Then either $z\in M(y)\cap F(x)$ or $z\in M(x)\cap F(y)$.
Let us suppose that the first condition holds (in the other case we have an analogous proof). Then, in turn, either $z'\in M(y)\cap F(x)$ or $z'\in M(x)\cap F(y)$.
In the first case we would have $z, z'\in M(y)$ and $z'$ a factor of $z$, a contradiction since $M(y)$ is antifactorial.
In the other case $z'\in M(x)$ and this is a contradiction since $z'\in F(z)\subseteq F(x)$.
\end{proof}

Note that, in general $M(x)\triangle M(y)$ is not antifactorial and $\Sigma^*(M(x)\triangle M(y))\Sigma^*$ is an ideal.

\begin{theorem}\label{th-base}
Let $x,y\in \Sigma^*$. Then $\D(x,y)$ is the base of the ideal $\Sigma^*(M(x)\triangle M(y))\Sigma^*$.
\end{theorem}
\begin{proof}
Since $\D(x,y)$ is antifactorial, in order to prove the statement it is sufficient to prove that any word $z\in (M(x)\triangle M(y))\backslash \D(x,y)$ has a factor $z'\in \D(x, y)$. Since $z\in M(x)\triangle M(y)$, then:
\begin{enumerate}
    \item\label{caso1} either $z\in M(x)$ and $z\not \in M(y)$;
    \item\label{caso2} or $z\in M(y)$ and $z\not \in M(x)$;
\end{enumerate}
Let us consider case \ref{caso1}. The proof of case \ref{caso2} is analogous.
By the hypothesis $z\not\in D(x \leftarrow y)\cup D(y \leftarrow x)$ we have $z\not\in M(x) \cap F(y)$. By conditions $z \in M(x)$, $z \notin M(y)$ and $z\not\in M(x) \cap F(y)$ we deduce that $z \not \in F(y)$. It means that $z$ is absent in $y$ but not minimal, therefore there exists a factorization $z=z_1z'z_2$ with $z'\in M(y)$. Since $z\in M(x)$, its proper factors are in $F(x)$. Therefore $z'\in F(x)\cap M(y)=D(x \leftarrow y)$.
\end{proof}

In other terms, Theorem \ref{th-base} states that considering $\D(x,y)$ is equivalent to ignore in $M(x)\triangle M(y)$ those words that have a proper factor in the same set.
We are now ready to define a distance based on the length weighted index applied to $\D(x,y)$:
$$\delta(x,y)=\mu (\D(x,y))=\sum_{w\in\D(x,y)}\frac{1}{|w|^2}=\sum_{w\in D(x\leftarrow y)}\frac{1}{|w|^2}+\sum_{w\in D(y\leftarrow x)}\frac{1}{|w|^2}.$$

We remark that as in the case of $\dist_l$, the distance $\delta$ takes into consideration elements among the shortest of  $M(x)\triangle M(y)$ because they are elements of the base of the ideal $\Sigma^*(M(x)\triangle M(y))\Sigma^*.$

We conclude the section with an example.
\begin{Example}
Let $x=cbaabdcb$ and $y=abcba$. As shown in Example \ref{ex-Dxy}, $\mathcal{D}(x,y)=\{d, aa, bc\}$. Then:
$$\delta(x,y)
=1+\frac{1}{2}
=\frac{3}{2}$$
\end{Example}



\section{Experimental analysis on biological sequences}

In this section we show results of our experimental study. The program to compute the two sample sets was implemented in the \texttt{C++} programming language and developed under GNU/Linux operating system. Our program makes use of the implementation of the compressed suffix tree available in the Succinct Data Structure Library~\cite{gbmp2014sea}. The input parameters are a (Multi)FASTA file with the input sequence(s). The implementation is distributed under the GNU General Public License, and it is available at \url{http://github.com/drjiagao/dmaw}.

\paragraph*{The data set.}

\begin{table}[tbh]
\caption{Coding sequences of the first exon sequences of $\beta$-globin genes from Human, Goat, Gallus, Opossum, Lemur, Mouse, Rabbit, Rat, Bovine, Gorilla, and Chimpanzee.}
\label{table:data}
\begin{center}
\begin{tabular} {|c|l|} 
\hline 
 \scriptsize SPECIES & \scriptsize CODING SEQUENCES\\ 
\hline
  \scriptsize Human & \vspace{-0.5 em} \scriptsize ATGGTGCACCTGACTCCTGAGGAGAAGTCTGCCGTTACTGCCCTGTGGGGCAA \\ 
 &   \scriptsize GGTGAACGTGGATTAAGTTGGTGGTGAGGCCCTGGGCAG \\
\hline
  \scriptsize Goat & \vspace{-0.5 em} \scriptsize ATGCTGACTGCTGAGGAGAAGGCTGCCGTCACCGGCTTCTGGGGCAAGGTGAA \\ 
 & \scriptsize AGTGGATGAAGTTGGTGCTGAGGCCCTGGGCAG\\
\hline
  \scriptsize Opossum & \vspace{-0.5 em} \scriptsize ATGGTGCACTTGACTTCTGAGGAGAAGAACTGCATCACTACCATCTGGTCTAA \\ 
 & \scriptsize GGTGCAGGTTGACCAGACTGGTGGTGAGGCCCTTGGCAG\\
\hline
  \scriptsize Gallus & \vspace{-0.5 em} \scriptsize ATGGTGCACTGGACTGCTGAGGAGAAGCAGCTCATCACCGGCCTCTGGGGCAA \\ 
  & \scriptsize GGTCAATGTGGCCGAATGTGGGGCCGAAGCCCTGGCCAG\\
               
\hline
  \scriptsize Lemur & \vspace{-0.5 em} \scriptsize ATGACTTTGCTGAGTGCTGAGGAGAATGCTCATGTCACCTCTCTGTGGGGCAA \\ 
 & \scriptsize GGTGGATGTAGAGAAAGTTGGTGGCGAGGCCTTGGGCAG\\
\hline
  \scriptsize Mouse & \vspace{-0.5 em} \scriptsize ATGGTTGCACCTGACTGATGCTGAGAAGTCTGCTGTCTCTTGCCTGTGGGCAA\\ 
 & \scriptsize AGGTGAACCCCGATGAAGTTGGTGGTGAGGCCCTGGGCAGG\\
\hline
  \scriptsize Rabbit & \vspace{-0.5 em} \scriptsize ATGGTGCATCTGTCCAGTGAGGAGAAGTCTGCGGTCACTGCCCTGTGGGGCAA \\ 
 & \scriptsize GGTGAATGTGGAAGAAGTTGGTGGTGAGGCCCTGGGC\\
\hline
  \scriptsize Rat & \vspace{-0.5 em} \scriptsize ATGGTGCACCTAACTGATGCTGAGAAGGCTACTGTTAGTGGCCTGTGGGGAAA \\ 
 & \scriptsize GGTGAACCCTGATAATGTTGGCGCTGAGGCCCTGGGCAG\\
\hline
  \scriptsize Gorilla & \vspace{-0.5 em} \scriptsize ATGGTGCACCTGACTCCTGAGGAGAAGTCTGCCGTTACTGCCCTGTGGGGCAA \\ 
 & \scriptsize GGTGAACGTGGATGAAGTTGGTGGTGAGGCCCTGGGCAGG\\
\hline
 \scriptsize Bovine & \vspace{-0.5 em} \scriptsize ATGCTGACTGCTGAGGAGAAGGCTGCCGTCACCGCCTTTTGGGGCAAGGTGAA \\ 
 & \scriptsize AGTGGATGAAGTTGGTGGTGAGGCCCTGGGCAG\\
\hline
  \scriptsize Chimpanzee & \vspace{-0.5 em} \scriptsize ATGGTGCACCTGACTCCTGAGGAGAAGTCTGCCGTTACTGCCCTGTGGGGCAA\\ 
 & \scriptsize GGTGAACGTGGATGAAGTTGGTGGTGAGGCCCTGGGCAGGTTGGTATCAAGG\\
\hline
\end{tabular}
\end{center}
\end{table}

The data set we take into consideration is the set of first exon sequences of $\beta$-globin genes from 11 species: Human, Goat, Gallus, Opossum, Lemur, Mouse, Rabbit, Rat, Bovine, Gorilla, and Chimpanzee. The gene family of $\beta$-globin is used to analyze DNA and its first exon is used for many DNA studies instead of computing the similarity/dissimilarity of the whole genomes. The length of such sequences is in fact quite short, on the order of 100 bases.

The coding sequences are listed in Table \ref{table:data}. This data set has been used by Liu and Wang  \cite{LIU2005}, by Chairungaee and Crochemore \cite{Croche2012} and  Rahman et al. \cite{RACR2016} to test the relative similarity/dissimilarity measures on these DNA sequences, in order to build phylogenetic trees. Here we test the new distance on the same data set in order to compare the resulting distance matrices and classifications with previous studies. In particular we consider the distance $\dist$ introduced in \cite{Croche2012}, since it is conceptually similar to $\delta$. The distance matrix of $\delta$ and $\dist$ are reported in Table \ref{table:ourdist} and \ref{table:triangle}, respectively.

\begin{table}[tbh]
\caption{Distance matrix obtained by applying our distance $\delta$ to the data set with 11 living organisms.}
\label{table:ourdist}
\begin{center}
\resizebox{\textwidth}{!}{\begin{tabular}  { | c | c | c | c | c | c | c | c | c | c | c | c |} 
\hline
SPECIE & Human & Goat & Oppssum & Gallus & Lemur & Mouse & Rabbit & Rat & Gorilla & Bovine & Chimpanzee\\ 
\hline
Human & 0.00 & 3.12 & 3.83 & 4.00 & 4.37 & 2.97 & 3.01 & 3.94 & 0.48 & 3.04 & 0.98 \\ 

\hline
Goat & & 0.00 & 3.80 & 3.27 & 3.34 & 3.47 & 3.39 & 3.45 & 2.88 & 1.04 & 3.13 \\ 

\hline
Opossum & & & 0.00 & 3.69 & 4.76 & 4.52 & 3.35 & 4.30 & 3.86 & 3.91 & 3.88 \\ 

\hline
Gallus & & & & 0.00 & 3.84 & 4.16 & 3.48 & 4.03 & 4.13 & 3.67 & 4.17 \\ 
\hline
Lemur & & & & & 0.00 & 4.23 & 4.31 & 4.45 & 4.16 & 3.31 & 4.60 \\ 

\hline
Mouse & & & & & & 0.00 & 3.71 & 3.27 & 2.60 & 3.23 & 2.91 \\ 

\hline
Rabbit & & & & & & & 0.00 & 3.53 & 3.06 & 3.45 & 3.21 \\ 

\hline
Rat & & & & & & & & 0.00 & 3.85 & 3.68 & 4.14 \\ 

\hline
Gorilla & & & & & & & & & 0.00 & 2.79 & 0.55 \\ 

\hline
Bovine & & & & & & & & & & 0.00 & 3.04 \\ 

\hline
Chimpanzee & & & & & & & & & & & 0.00 \\ 

\hline
\end{tabular}}
\end{center}
\end{table}

\begin{table}[t]
\caption{Distance matrix obtained by applying the distance $\dist$ to the data set with 11 living organisms.}
\label{table:triangle}
\begin{center}
    \resizebox{\textwidth}{!}{\begin{tabular}  { | c | c | c | c | c | c | c | c | c | c | c | c |} 
\hline
SPECIE & Human & Goat & Oppssum & Gallus & Lemur & Mouse & Rabbit & Rat & Gorilla & Bovine & Chimpanzee\\ 
\hline
Human & 0.00 & 8.34 & 10.55 & 11.18 & 10.61 & 8.21 & 8.44 & 10.54 & 1.09 & 8.22 & 2.65 \\ 
\hline
Goat & & 0.00 & 10.37 & 9.23 & 9.00 & 8.50 & 8.63 & 10.21 & 8.06 & 2.78 & 8.83 \\ 

\hline
Opossum & & & 0.00 & 11.01 & 12.34 & 12.60 & 10.24 & 12.08 & 10.45 & 10.87 & 11.00 \\ 

\hline
Gallus & & & & 0.00 & 11.08 & 10.93 & 9.17 & 12.55 & 11.46 & 10.22 & 11.76 \\ 

\hline
Lemur & & & & & 0.00 & 10.69 & 10.93 & 11.29 & 10.24 & 8.76 & 11.11 \\ 

\hline
Mouse & & & & & & 0.00 & 10.11 & 10.32 & 7.43 & 8.11 & 8.50 \\ 

\hline
Rabbit & & & & & & & 0.00 & 10.53 & 8.58 & 9.06 & 9.00 \\ 

\hline
Rat & & & & & & & & 0.00 & 10.60 & 10.41 & 11.71 \\ 

\hline
Gorilla & & & & & & & & & 0.00 & 7.93 & 1.64 \\ 

\hline
Bovine & & & & & & & & & & 0.00 & 8.71 \\ 

\hline
Chimpanzee & & & & & & & & & & & 0.00 \\ 

\hline
\end{tabular}}
\end{center}
\end{table}

\paragraph*{Experimental results of cardinality and total length.}
 
In previous sections we have remarked that in order to compute the dissimilarity of two sequences $x$ and $y$, we chose to use the measure $\mu$ on the elements of $\D(x,y)$, i.e.  minimal absent words of one sequence that occur in the other one. In our opinion this set considers the elements in $M(x)\triangle M(y)$ that really distinguish the words $x$ and $y$. As a side effect, we state that this set in general contains a compressed information since, by construction, this set is expected to be smaller than $M(x)\triangle M(y)$. In what follows, we try to experimentally extimate how smaller it is. We consider two parameters as indices of the size of a set: the cardinality and the total length.

In Table \ref{table:rationumberMAW} we report the matrix of the values of the ratio $$Card(\D(x,y)) \slash Card(M(x) \triangle M(y))$$ for each pair $(x,y)$ of sequences in the dataset. As we can observe, in practice, this ratio is approximately $1 \slash 4$. This suggests that $\D(x,y)$ contains around $25\%$ of the elements of $M(x) \triangle M(y)$.

In ordet to clarify the reason for this smaller cardinality, we consider the following example.
 
 \begin{Example}\label{ex:ratmouse}
 Let $x$ and $y$ be the coding sequences of Mouse and Rat, respectively. The set $M(x) \triangle M(y)$ contains $210$ elements whereas $\D(x,y)$ contains only $47$ elements: hence the remaining $163$ elements are minimal absent for a sequence and absent, although not minimal, for the other one. Therefore, in any case, they are absent for both. For instance, the word $CCCCC$ is minimal absent for the Mouse; the Rat has not five consecutive $C$'s because it has not four consecutive $C$'s ($CCCC$ is minimal absent for the Rat). Hence, in our opinion, the significant information is that $CCCC$ is minimal absent for the Rat but is a factor for the Mouse.
 \end{Example}

\begin{table}[thp!]
\caption{Matrix of the values $Card(\D(x,y)) \slash Card(M(x) \triangle M(y))$ for the data set with 11 living organisms.}
\label{table:rationumberMAW}
\begin{center}
\resizebox{\textwidth}{!}{\begin{tabular} { | c | c | c | c | c | c | c | c | c | c | c | c |} 
\hline
SPECIE & Human & Goat & Oppssum & Gallus & Lemur & Mouse & Rabbit & Rat & Gorilla & Bovine & Chimpanzee\\ 
\hline
Human & 0.00 & 0.27 & 0.25 & 0.25 & 0.27 & 0.26 & 0.24 & 0.26 & 0.30 & 0.25 & 0.24 \\ 

\hline
Goat & & 0.00 & 0.26 & 0.24 & 0.27 & 0.31 & 0.27 & 0.24 & 0.26 & 0.26 & 0.25 \\ 

\hline
Opossum & & & 0.00 & 0.23 & 0.26 & 0.26 & 0.23 & 0.24 & 0.25 & 0.24 & 0.24 \\ 

\hline
Gallus & & & & 0.00 & 0.25 & 0.27 & 0.27 & 0.22 & 0.25 & 0.24 & 0.25 \\ 

\hline
Lemur & & & & & 0.00 & 0.30 & 0.29 & 0.29 & 0.27 & 0.28 & 0.28 \\ 

\hline
Mouse & & & & & & 0.00 & 0.24 & 0.22 & 0.25 & 0.29 & 0.24 \\ 

\hline
Rabbit & & & & & & & 0.00 & 0.23 & 0.24 & 0.25 & 0.24 \\ 

\hline
Rat & & & & & & & & 0.00 & 0.25 & 0.26 & 0.24 \\ 

\hline
Gorilla & & & & & & & & & 0.00 & 0.23 & 0.21 \\ 

\hline
Bovine & & & & & & & & & & 0.00 & 0.23 \\ 

\hline
Chimpanzee & & & & & & & & & & & 0.00 \\ 

\hline

\end{tabular}}
\end{center}
\end{table}

Furthermore, in Table \ref{table:rationlenghtMAW} we consider the ratio of the total lengths  $$s(\D(x,y)) \slash s(M(x)\triangle M(y))$$ for each pair $(x,y)$ of sequences in the data set. As one can notice, in all the cases this ratio is even less than the ratio of cardinalities, in fact $\D(x,y)$ has around the $20\%$ of the total length of $M(x)\triangle M(y)$, showing that not only the set $\D(x,y)$  has fewer elements than $M(x)\triangle M(y)$, but also these elements are among the shortest. Hence, handling the set $D(x,y)$ is more convenient because its size is $1/5$ of the size of $M(x)\triangle M(y)$.

 \begin{Example}  
 Let $x$ and $y$ be the coding sequences of Mouse and Rat, respectively. The average length of the words in $M(x) \triangle M(y)$ is about $5$ whereas the average length of the words in $\D(x,y)$ is about $4$. The total length $s(\D(x,y))$ is $199$ whereas $s(M(x) \triangle M(y))$ is $1038$, which witness a great reduction of the information needed for computing the distance.
 \end{Example}

\begin{table}[thp!]
\caption{Matrix of the values $s(D(x,y)) \slash s(M(x)\triangle M(y))$, for the data set with 11 living organisms}
\label{table:rationlenghtMAW}
\begin{center}
\resizebox{\textwidth}{!}{\begin{tabular} { | c | c | c | c | c | c | c | c | c | c | c | c |} 
\hline
SPECIE & Human & Goat & Oppssum & Gallus & Lemur & Mouse & Rabbit & Rat & Gorilla & Bovine & Chimpanzee\\ 
\hline
Human & 0.00 & 0.23 & 0.20 & 0.20 & 0.21 & 0.22 & 0.19 & 0.21 & 0.25 & 0.21 & 0.19 \\ 

\hline
Goat & & 0.00 & 0.22 & 0.19 & 0.23 & 0.26 & 0.23 & 0.20 & 0.22 & 0.21 & 0.21 \\ 

\hline
Opossum & & & 0.00 & 0.19 & 0.20 & 0.22 & 0.20 & 0.20 & 0.20 & 0.19 & 0.19 \\ 

\hline
Gallus & & & & 0.00 & 0.21 & 0.23 & 0.23 & 0.19 & 0.21 & 0.20 & 0.21 \\ 

\hline
Lemur & & & & & 0.00 & 0.26 & 0.24 & 0.25 & 0.21 & 0.24 & 0.22 \\ 

\hline
Mouse & & & & & & 0.00 & 0.17 & 0.19 & 0.21 & 0.25 & 0.21 \\ 

\hline
Rabbit & & & & & & & 0.00 & 0.19 & 0.20 & 0.20 & 0.20 \\ 

\hline
Rat & & & & & & & & 0.00 & 0.21 & 0.23 & 0.20 \\ 

\hline
Gorilla & & & & & & & & & 0.00 & 0.19 & 0.16 \\ 

\hline
Bovine & & & & & & & & & & 0.00 & 0.18 \\ 

\hline
Chimpanzee & & & & & & & & & & & 0.00 \\ 

\hline

\end{tabular}}
\end{center}
\end{table}

\paragraph*{Experimental results of phylogeny building using the two distances.}
In the figures shown in this section we give the phylogenetic trees of the 11 species computed using NJ and UPGMA methods on the two distances $\delta$ and $\dist$ considered in this paper.
Based on this data set, we can verify the goodness of our method by checking which of the following statements, based on the real biological phenomena, are verified by our new distance (cf. \cite{LIU2005}):
\begin{itemize}
\item $[REL\, 1]$ Gorilla and Chimpanzee are most similar to Human;
\item $[REL\, 2]$ Goat and Bovine should be similar to each other;
\item $[REL\, 3]$ Rat and Mouse should be similar to each other;
\item $[REL\, 4]$ Gallus should be remote from the other species because Gallus is the only non-mammalian representative in this group;
\item $[REL\, 5]$ Opossum is the most remote species from the remaining mammals;
\item $[REL\, 6]$ Besides Gallus and Opossum, Lemur is more remote from the other species relatively.
\end{itemize}

As we can observe from the phylogenetic trees that outputs from the matrix of distance $\delta$, relations 1 and 2 are fully satisfied. Rat and Mouse are quite close to each other, especially in NJ. 
Gallus, Opossum and Lemur are the farthest species from the others.


\begin{figure}[thp]
\centering
\begin{tikzpicture}[every tree node/.style={font=\small},
level distance=0.8cm,sibling distance=0.01cm, 
edge from parent path={(\tikzparentnode.south) -- +(0,-8pt) -| (\tikzchildnode)}],
frontier/.style={distance from root=350pt} 
\Tree 
[
\edge node[near end,left] {}; 
[
\edge node[] {}; [.Opossum ]
\edge node[] {}; [.Rabbit ]
]
\edge node[near end,right] {}; 
[ 
\edge node[] {}; [.Gallus ]
\edge node[] {}; 
[
\edge node[] {}; 
[
\edge node[] {}; [.Lemur ]
\edge node[] {}; 
[
\edge node[] {}; [.Goat ]
\edge node[] {}; [.Bovine ]
]
]
\edge node[] {}; 
[
\edge node[] {}; [.Rat ]
\edge node[] {}; 
[
\edge node[] {}; [.Mouse ]
\edge node[] {}; 
[
\edge node[] {}; [.Human ]
\edge node[] {}; 
[
\edge node[] {}; [.Gorilla ]
\edge node[] {}; [.Chimpanzee ]
]
]
]
]
]
]
]
\end{tikzpicture}
\caption{The phylogenetic tree of the 11 species computed using Neighbor Joining applied to the matrix of distance $\delta$ shown in Table \ref{table:ourdist}.}
\label{fig:CastiglioneNJ}
\end{figure}
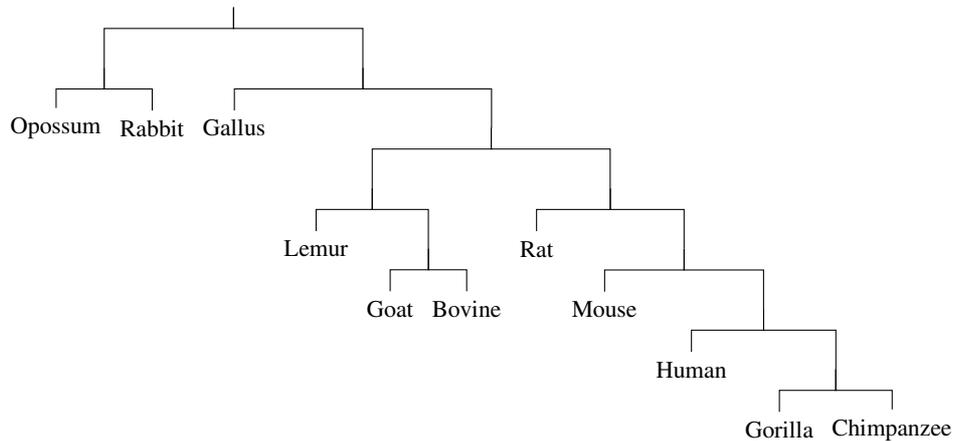

\begin{figure}[thp]
\centering
\begin{tikzpicture}[every tree node/.style={font=\small},
level distance=0.8cm,sibling distance=0.01cm, 
edge from parent path={(\tikzparentnode.south) -- +(0,-8pt) -| (\tikzchildnode)}],
frontier/.style={distance from root=350pt} 
\Tree 
[
\edge node[near end,left] {};
[
\edge node[] {}; [.Rat ]
\edge node[] {}; 
[
\edge node[] {}; [.Opossum ]
\edge node[] {}; 
[
\edge node[] {}; [.Gallus ]
\edge node[] {}; [.Rabbit ]
]
]
]
\edge node[near end,right] {};
[
\edge node[] {}; 
[
\edge node[] {}; [.Mouse ]
\edge node[] {}; 
[
\edge node[] {}; [.Human ]
\edge node[] {}; 
[
\edge node[] {}; [.Gorilla ]
\edge node[] {}; [.Chimpanzee ]
]
]
]
\edge node[] {}; 
[
\edge node[] {}; [.Lemur ]
\edge node[] {}; 
[
\edge node[] {}; [.Goat ]
\edge node[] {}; [.Bovine ]
]
]
]
]
\end{tikzpicture}
\caption{The phylogenetic tree of the 11 species computed using Neighbor Joining applied to the matrix of the distance $\dist$ shown in Table \ref{table:triangle}.}
\label{fig:CrochemoreNJ}
\end{figure}
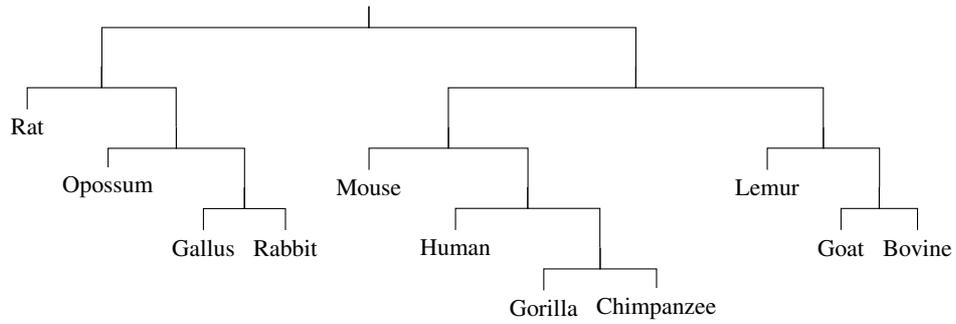

\begin{figure}
\centering
\begin{tikzpicture}[every tree node/.style={font=\small},
level distance=0.8cm,sibling distance=0.01cm, 
edge from parent path={(\tikzparentnode.south) -- +(0,-8pt) -| (\tikzchildnode)}],
frontier/.style={distance from root=350pt} 
\Tree 
[
\edge node[near end,left] {}; [.Lemur ]
\edge node[near end,right] {}; 
[
\edge node[] {}; 
[
\edge node[] {}; [.Opossum ]
\edge node[] {}; [.Gallus ]
]
\edge node[] {}; 
[
\edge node[] {}; [.Rat ]
\edge node[] {}; 
[
\edge node[] {}; [.Rabbit ]
\edge node[] {}; 
[
\edge node[] {}; 
[
\edge node[] {}; [.Goat ]
\edge node[] {}; [.Bovine ]
]
\edge node[] {}; 
[
\edge node[] {}; [.Mouse ]
\edge node[] {}; 
[
\edge node[] {}; [.Chimpanzee ]
\edge node[] {}; 
[
\edge node[] {}; [.Human ]
\edge node[] {}; [.Gorilla ]
]
]
]
]
]
]
]
]
\end{tikzpicture}
\caption{The phylogenetic tree of the 11 species computed using UPGMA applied to the matrix of distance $\delta$ shown in Table \ref{table:ourdist}.}
\label{fig:CastiglioneUPGMA}
\end{figure}
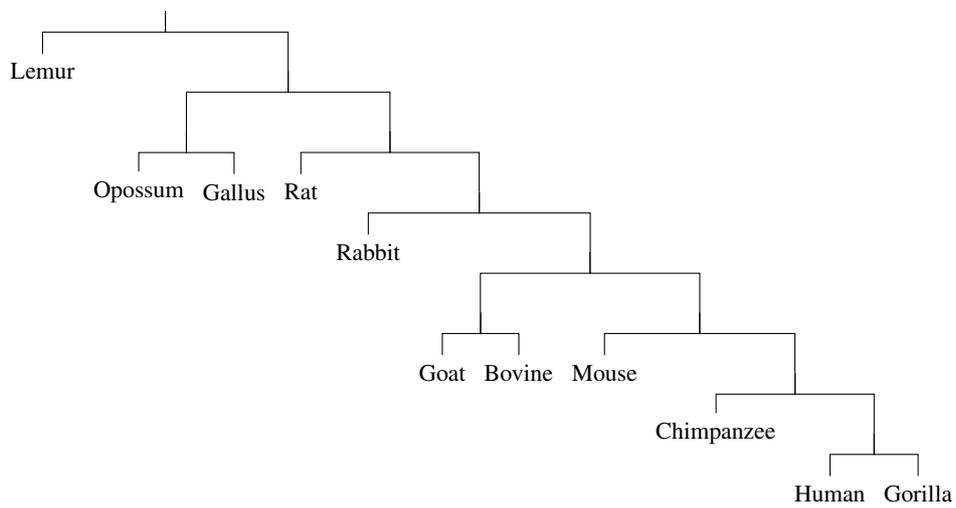

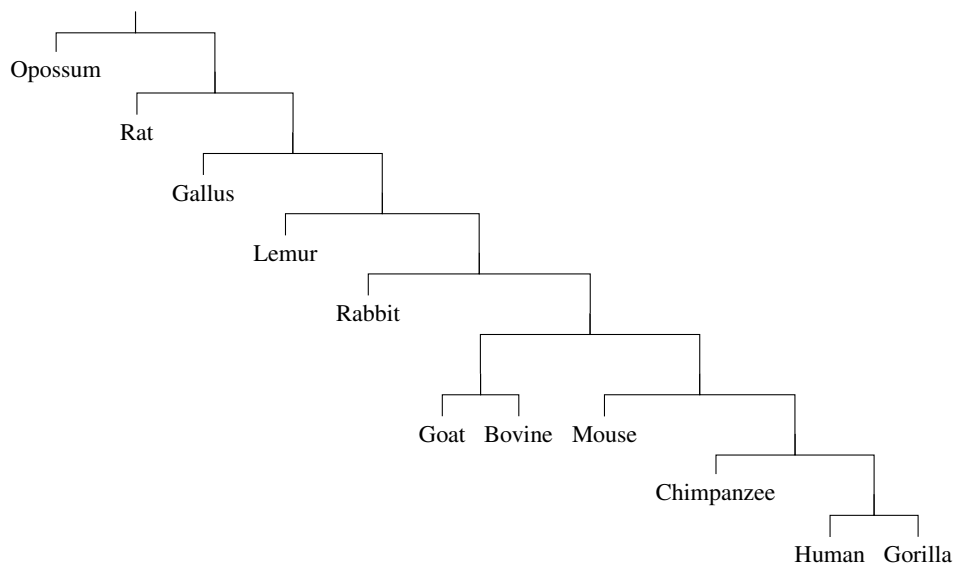
\begin{figure}
\centering
\begin{tikzpicture}[every tree node/.style={font=\small},
level distance=0.8cm,sibling distance=0.01cm, 
edge from parent path={(\tikzparentnode.south) -- +(0,-8pt) -| (\tikzchildnode)}],
frontier/.style={distance from root=350pt} 
\Tree 
[
\edge node[near end,left] {}; [.Opossum ]
\edge node[near end,right] {}; 
[
\edge node[] {}; [.Rat ]
\edge node[] {}; 
[
\edge node[] {}; [.Gallus ]
\edge node[] {}; 
[
\edge node[] {}; [.Lemur ]
\edge node[] {}; 
[
\edge node[] {}; [.Rabbit ]
\edge node[] {}; 
[
\edge node[] {}; 
[
\edge node[] {}; [.Goat ]
\edge node[] {}; [.Bovine ]
]
\edge node[] {}; 
[
\edge node[] {}; [.Mouse ]
\edge node[] {}; 
[
\edge node[] {}; [.Chimpanzee ]
\edge node[] {}; 
[
\edge node[] {}; [.Human ]
\edge node[] {}; [.Gorilla ]
]
]
]
]
]
]
]
]
]
\end{tikzpicture}
\caption{The phylogenetic tree of the 11 species computed using UPGMA applied to the matrix of distance $\dist$  shown in Table \ref{table:triangle}.}
\label{fig:CrochemoreUPGMA}
\end{figure}

\newpage

\section{Conclusions.}
In this paper we presented a new similarity/dissimilarity measure for pairs of sequences $(x,y)$, that takes inspiration from the one, based on the symmetric difference of the sets of minimal absent words of $x$ and $y$, introduced in ~\cite{Croche2012}. We consider instead a subset of the symmetric difference, consisting of the set of the minimal abstent words of $x$ that occur in $y$ and viceversa. This is motivated by the observation that the symmetric difference contains sequences that are absent in both $x$ and $y$, although minimal in just one of them, and we think that these words should not be considered to distinguish $x$ and $y$.
In order to validate the goodness of this dissimilarity measure, we have tested it on the data set of the first exon sequences of the $\beta$-globine gene of 11 living organisms, that is commonly considered in papers dealing with distances. By constructing the relative philogenetic trees (both with NJ and UPGMA methods), the distance seems to satisfy most of the requirements for a "good" distance on this data set.

We also construct two tables where we consider, for all the pairs of  species in the data set, the ratio between the cardinalities and the total length (respectively) of the sets $\D(x,y)$ and $M(x)\triangle M(y)$, showing that the set $D(x,y)$ is in general much smaller than $M(x)\triangle M(y)$ according to both of these parameters, with potential improvement to computational issues.
\bibliography{ref}

\end{document}